\documentclass[pra, twocolumn]{revtex4-1}
\usepackage{amsmath}
\usepackage{amsfonts}
\usepackage{amssymb}
\usepackage{amsthm}
\usepackage{color}
\usepackage[normalem]{ulem}

\newcommand{\sset}{\mathcal{S}}

\newcommand{\unit}{\mathcal{U}}

\newcommand{\innprod}[2]{\langle #1 , #2 \rangle}
\newcommand{\trans}{T^{\dagger}}
\newcommand{\condprob}[2]{p\, (#1 | #2 )}
\newcommand{\alp}{\mathcal{A}}

\newcommand{\subuniteffsati}{\mathcal{S}_{\{i\}}}

\newtheorem{theorem}{Theorem}
\newtheorem{lemma}{Lemma}

\newtheorem{corollary}{Corollary}

\begin{document}

\title{Reversible Dynamics in Strongly Non Local Boxworld Systems}

\author{Sabri W. Al-Safi} \email{S.W.Al-Safi@damtp.cam.ac.uk}
\affiliation{DAMTP, Centre for Mathematical Sciences, Wilberforce Road, Cambridge CB3 0WA, UK}

\author{Anthony J. Short} \email{tony.short@bristol.ac.uk}
\affiliation{H. H. Wills Physics Laboratory, University of Bristol, Tyndall Avenue, Bristol, BS8 1TL, UK}

\begin{abstract}  
In order to better understand the structure of quantum theory, or speculate about theories that may supercede it, it can be helpful to consider alternative physical theories. ``Boxworld'' describes one such theory, in which all non-signaling correlations are achievable.  In a limited class of multipartite Boxworld systems - wherein all subsystems are identical and all measurements have the same number of outcomes - it has been demonstrated that the set of reversible dynamics is `trivial', generated solely by local relabellings and permutations of subsystems. We develop the convex formalism of Boxworld to give an alternative proof of this result, then extend this proof to all multipartite Boxworld systems, and discuss the potential relevance to other  theories. These results lend further support to the idea that the rich reversible dynamics in quantum theory may be the key to understanding its structure and its informational capabilities.
\end{abstract}

\maketitle

\section{Introduction.}  To gain a better understanding of quantum theory, and to  explore possible future modifications, it can be helpful to view quantum theory from the `outside' -- as one member of a broader class of  physical theories. One such approach is to consider the class of general probabilistic theories \cite{barrett05, hardy01, chiribella10, chiribella11}, which are based on operational notions that allow many different theories to be represented using the same  intuitive mathematical formalism. 

A natural alternative  to quantum theory within this class is \emph{Boxworld} \cite{short&barrett10} (originally called generalised non-signalling theory in \cite{barrett05}). Like quantum theory, this theory admits non-local correlations which cannot be explained by any locally realistic model \cite{bell65}. In fact it admits all non-signaling correlations, including those which maximally violate Bell inequalities, such as the well-known \emph{PR-Box} \cite{pr94}. These super-strong correlations cause Boxworld to differ markedly from the world we observe: for example, any distributed computation could be performed with the transmission of a single bit \cite{vandam05}, and bit-commitment would be possible without using relativistic effects \cite{buhrman06}. 

In attempting to  better understand what makes quantum theory uniquely successful at describing the world, attention has been given to various principles which we expect nature to obey, but which are often violated by other probabilistic theories. A common example of this is \emph{reversibility} (or transitivity) \cite{hardy01, masanes12, muller12, chiribella12, masanes11}, which demands that any two pure states are linked by a reversible transformation. 

In \cite{colbeck10}, Colbeck \emph{et al.} proved that, if one assumes that all subsystems are identical and all measurements have the same number of outcomes, the set of reversible Boxworld multipartite dynamics is generated by local operations and permutations of systems. Such Boxworld systems  cannot experience reversible interactions and thus violate reversibility. In particular, a PR-box (or any other entangled state) could not generated reversibly from an initial product state. 

Other recent results have also highlighted the importance of reversibility. It has been shown that two systems whose state-spaces are $d$-dimensional balls can only interact in a continuous and reversible way if $d=3$ (in which case the systems correspond to qubits in the Bloch-sphere representation) \cite{muller11}. Furthermore, any theory in which  local systems are identical to qubits and in which there exists at least one continuous reversible interaction must globally be identical to quantum theory \cite{masanes12}. Generalisations of such results may have great significance in explaining why our world looks quantum, or in finding theories which may potentially supersede quantum theory.

 A common feature of the reversibility results cited above \cite{colbeck10, masanes12, muller11} is that they consider interactions between systems which are locally identical. However, the reversibility of quantum theory carries over to  the case where, for example, the systems have differing Hilbert space dimension. Could there exist non-trivial reversible interactions between different types of system in these other general probabilistic theories? In this paper we provide an alternative proof of Colbeck \emph{et al.}'s result and then extend it to apply to any combination of Boxworld systems, including different types (so long as none of the systems are classical).

The structure of the paper is as follows: in \S 2 we outline the convex formalism of Boxworld and introduce some extra terminology useful for our exposition. In \S 3 we prove some results about the structure of this convex set, which we show in \S 4 are sufficient to recover the result of \cite{colbeck10}. We then extend this proof to the general case of non-identical systems. 

\section{ Set-up and notation. }
In order to compare quantum theory to alternative theories, such as Boxworld, it is helpful to define a mathematical framework which is broad enough to describe any such theory. Here we consider an operational framework for defining general probabilistic theories \cite{barrett05, hardy01} in which the \emph{state} of a system is specified by the probabilities it assigns to   \emph{effects}, or measurement outcomes. In particular, for each system we will  assume there exists a  finite set of \emph{fiducial} measurements, which are sufficient to deduce the outcome probabilities of any other measurement (e.g. for a qubit we could take measurements of the Pauli operators $\sigma_x$, $\sigma_y$ and $\sigma_z$ as the fiducial set). 

Consider a joint system composed of $N$ individual systems. One possible global measurement on the joint system involves a local fiducial measurement being performed on each system. The value $\condprob{a_1 , \ldots , a_N}{x_1 , \ldots , x_N}$ then denotes the probability of outcomes $a_1, \ldots , a_N$ occurring at systems $1, \ldots , N$ respectively, given that the measurement choices $x_1, \ldots , x_N$ were made on those systems. We assume that a complete specification of the values $\condprob{a_1 , \ldots , a_N}{x_1 , \ldots , x_N}$ is sufficient to determine the state of the joint system (an assumption commonly known as \emph{local tomography} \cite{hardy01, barrett05}). These values obey the normal laws of probability:
\begin{equation}\label{statepositivity}
0 \leq \condprob{a_1 , \ldots , a_N}{x_1 , \ldots , x_N} \leq 1 ,
\end{equation}
and for any fixed choice of the $x_i$: 
\begin{equation}\label{statenormalization}
\sum_{a_1, \ldots , a_N} \condprob{a_1 , \ldots , a_N}{x_1 , \ldots , x_N} = 1.
\end{equation}

To ensure that information cannot be sent between the systems (since, for example, they may be spatially separated), we demand also that the \emph{no-signaling} condition is satisfied by the joint distribution $p$. This condition says that the outcome statistics for any subset $\Omega$ of the $N$ systems must not be affected by measurement choices made on systems not in $\Omega$, i.e. there is no way for one system to signal information to any other systems. In mathematical terms, we demand that for all $\Omega \subset [N]$, the \emph{marginal distribution}
\begin{equation}\label{statenosignaling}
\sum_{a_i : i\notin \Omega} \condprob{a_1 , \ldots , a_N}{x_1 , \ldots , x_N}
\end{equation}
is well-defined, independent of the value of $x_i$ for systems $i \notin\Omega$.

 \emph{Boxworld} consists of all multipartite states whose joint outcome distributions for the fiducial measurements obey \eqref{statepositivity}, \eqref{statenormalization}, and \eqref{statenosignaling}.  The allowed measurements and transformations in Boxworld are all those which are well-defined within the operational framework.

Instead of considering the probability distribution $\condprob{a_1 , \ldots , a_N}{x_1 , \ldots , x_N}$ directly, it is often convenient to represent general probabilistic theories using real vector spaces, in which states $s$ and effects $e$ are specified by vectors such that $\innprod{e}{s}$ equals the probability of effect $e$ occuring for a system in state $s$. Let there be $M^{(i)}$ measurement choices on system $i$, and $K_{j}^{(i)}$ outcomes for the $j$th measurement on system $i$. When $M^{(i)} = 1$, a single probability distribution is sufficient to describe the state of the system, and we say the system is classical  (our results apply exclusively to the case where all systems are non-classical). The effect vectors for system $i$ may be constructed as follows: pick a linearly independent set of vectors $\{\unit^{(i)}, X^{(i)}_{a_i|x_i}\} \subset \mathbb{R}^d$ for $1\leq x_i \leq M^{(i)}$ and $1\leq a_i \leq K^{(i)}_{x_i}-1$, where $d=1+\sum_{x_i=1}^{M^{(i)}} (K^{(i)}_{x_i}-1)$. $\unit^{(i)}$ represents the \emph{unit effect}: the unique effect for which any allowed (normalized) state gives probability $1$. $X^{(i)}_{a_i|x_i}$ is the fiducial effect corresponding to measuring $x_i$ and obtaining outcome $a_i$. The remaining fiducial effect vectors are defined $X^{(i)}_{K^{(i)}_{x_i}|x_i} = \unit^{(i)} - \sum_{a_i=1}^{K^{(i)}_{x_i}-1} X^{(i)}_{a_i|x_i}$.

It turns out that the tensor product of the vector spaces characterizing each individual system provides a neat representation of states and effects in joint Boxworld systems \cite{barrett05}. Suppose system $i$ is represented by a real vector space $\mathbb{R}^{d_i}$ as above, and let $\mathbb{R}_N = \mathbb{R}^{d_1}\otimes \cdots \otimes \mathbb{R}^{d_N}$. The N-partite fiducial effects are defined to be the vectors of the form $X_{a_1|x_1}^{(1)}\otimes \cdots \otimes X_{a_N|x_N}^{(N)}$, where $X_{a_i|x_i}^{(i)}$ is a fiducial effect on system $i$. The $N$-partite unit effect is defined by $\unit = \unit^{(1)}\otimes \cdots \otimes \unit^{(N)}$.

Any Boxworld state whose measurement statistics obey \eqref{statepositivity}, \eqref{statenormalization}, and \eqref{statenosignaling} corresponds to a unique vector $s \in \mathbb{R}_N$, such that $\innprod{\unit}{s}=1$ (i.e. the state is normalized) and $p(a_1, \ldots , a_N | x_1 , \ldots , x_N) = \innprod{X_{a_1|x_1}^{(1)}\otimes \cdots \otimes X_{a_N|x_N}^{(N)}}{s}$ \cite{colbeck10}. Let the set of allowed state vectors be denoted by $\sset \subset \mathbb{R}_N$. \emph{Pure product states} are those of the form $s^{(1)}\otimes \cdots \otimes s^{(N)}$ where $s^{(i)}$ deterministically assigns to each fiducial measurement $x_i$ on system $i$, a definite outcome $1 \leq s^{(i)}_{x_i} \leq K^{(i)}_{x_i}$, i.e. $\innprod{X_{a_i|x_i}^{(i)}}{s^{(i)}} = 1$ iff $s^{(i)}_{x_i} = a_i$.

The allowed $N$-partite effects in Boxworld are all vectors $e \in \mathbb{R}_N$ such that  $\innprod{e}{s} \in [0,1]$ for all $s \in \sset$. If $\innprod{e}{s} = 1$, we will say that the state $s$ \emph{hits} the effect $e$. We will be particularly interested in effects of the form $E = \sum_{\alpha} e_{\alpha}$, where each $e_{\alpha}$ is a fiducial effect.  We say in this case that $\{e_{\alpha}\}$ forms a \emph{decomposition} of $E$, or $E$ admits the decomposition $\{e_{\alpha}\}$. We will tend to use lowercase letters for fiducial effects, and uppercase for sums of extreme ray effects.

An effect $E$ is \emph{multiform} if it can be written $E=\sum_{\alpha} e_{\alpha} = \sum_{\beta} f_{\beta}$ where $\{e_{\alpha}\}$ and $\{f_{\beta}\}$ are distinct sets of fiducial effects. Effects of the form $X_{a_1|x_1}^{(1)}\otimes \cdots \otimes \unit^{(i)} \otimes \cdots \otimes X_{a_N|x_N}^{(N)}$, where exactly one component of the tensor product is the unit effect, and the remainder are fiducial effects, are said to be \emph{sub-unit} effects, or an $i$-\emph{sub-unit} effect if the $i$th component is the unit effect. For each $x_i$, $\unit^{(i)}$ has a distinct decomposition $\sum_{a_i=1}^{K^{(i)}_{x_i}} X^{(i)}_{a_i|x_i}$, hence sub-unit effects are trivially multiform (as we assume all systems are non-classical and hence have at least two fiducial measurements).

Fiducial and sub-unit effects are tensor products of vectors, so it makes sense to refer to their $i$th component, e.g. $E^{(i)} = X_{a_i|x_i}^{(i)}$. For a subset $\Omega \subseteq [N]$ we will write $E^{\Omega} = \bigotimes_{i\in \Omega} E^{(i)}$, e.g. $E^{\{1,3\}} = X_{a_1|x_1}^{(1)}\otimes X_{a_3|x_3}^{(3)}$.

Finally, we say that a set of fiducial  effects $\{e_{\alpha}\}_{\alpha \in A} $ (strictly) \emph{covers} the effect $E$ if there is some (strict) subset $B\subset A$ such that $\sum_{\alpha \in B} e_{\alpha} = E$.


\section{ Decompositions.}  We now prove some results concerning multiform effects. Given that none of the systems are classical, the simplest multiform effects are the sub-unit effects, which have various decompositions according to the different measurement choices on the system whose component is the unit effect. The following Lemma shows that these are the only possible decompositions of a sub-unit effect. 

\begin{lemma} \label{redefflemma}
Let $E = \sum_{\alpha} e_{\alpha}$ be an $i$-sub-unit effect. Then each fiducial effect $e_{\alpha}$ satisfies $e_{\alpha}^{(j)} = E^{(j)}$ for all components $j \neq i$. Moreover, the set of $i$th components $\{e^{(i)}_{\alpha}\}$ forms a fiducial measurement on system $i$.
\end{lemma}

\begin{proof} See appendix \ref{app_lem1} \end{proof} 

\begin{corollary} \label{redeffcor}
Let $E=\sum_{\alpha=1}^r e_{\alpha}$ be a sub-unit effect. Then $\{e_{\alpha}\}$ does not strictly cover a multiform effect. 
\end{corollary}

\begin{proof}Suppose without loss of generality  that for $s<r$ we have $\sum_{\alpha=1}^{s} e_{\alpha}= \sum_{\beta=1}^{t} f_{\beta}$. Then $\{f_1, \ldots , f_t, e_{s+1}, \ldots e_r\}$ is a decomposition of $E$ containing $e_r$. However, it follows from Lemma \ref{redefflemma} that there is a unique decomposition of $E$ containing $e_r$, hence $\{e_{\alpha}\}_{{\alpha}=1}^{s} = \{f_{\beta}\}_{{\beta}=1}^{t}$.
\end{proof}

For convenience we will assume from here on that the systems are arranged in order of increasing numbers of measurement outcomes, i.e. $K_{j}^{(i)} \leq K_{j+1}^{(i)}$ and $K_{1}^{(i)} \leq K_{1}^{(i+1)}$; this amounts to no more than a relabelling of systems and measurement choices. $K^{(1)}_1$ is therefore the smallest number of outcomes possible for any fiducial measurement. The following lemma restricts the type of effects which can have small decompositions.

\begin{lemma} \label{mindecomplemma}
For $r\leq K^{(1)}_1$ suppose that $\{e_{\alpha}\}_{\alpha=1}^{r}$does not cover any sub-unit effects. Then for any fiducial effect $f \notin \{e_{\alpha}\}$, there is a pure product state which hits $f$ but none of the $e_{\alpha}$.
\end{lemma}

\begin{proof} See appendix \ref{app_lem2} \end{proof} 

\begin{corollary}\label{mindecompcor}
The only multiform effects which have a decomposition with exactly $K^{(1)}_1$ elements are sub-unit effects.
\end{corollary} 
\begin{proof}
Suppose $E=\sum_{\alpha=1}^{r} e_{\alpha} = \sum_{\beta=1}^{s} f_{\beta}$ are distinct decompositions, with $r = K^{(1)}_1$, and suppose without loss of generality that $f_1 \notin \{e_{\alpha}\}_{{\alpha}=1}^r$. Every pure product state which hits $f_1$ must also hit one of the $e_{\alpha}$, so it follows from Lemma \ref{mindecomplemma} that  $\{e_{\alpha}\}$ covers a sub-unit effect. By Lemma \ref{redefflemma}, every decomposition of a sub-unit effect has at least $K^{(1)}_1$ elements, hence $E$ is itself a sub-unit effect.
\end{proof}

\section{Dynamics } As well as states and measurements, a physical theory must have some notion of dynamics which transform the state of a system. We call $T$ an allowed transformation if it is a convex-linear mapping of $\sset$ to itself. i.e. for $s_1, s_2 \in  \sset$ and $0 \leq p \leq 1$, $T(p s_1 + (1-p) s_2) = pT(s_1) + (1-p)T(s_2)$. Any map $T$ satisfying this condition can be extended to a linear map on $\mathbb{R}_N$ \cite{barrett05}, so we will assume that $T$ is linear.

Operationally, a transformation is determined by how it affects outcome probabilities. Since $\innprod{e}{T(s)} = \innprod{T^{\dagger}(e)}{s}$, any transformation $T$ may equivalently be defined by how its adjoint linear map $T^{\dagger}$ acts on effects. From here on we will use the latter perspective.

$\trans$ is \emph{reversible} if there exists an allowed transformation $S$ such that $S^{\dagger}\trans(e) = e$ for all effects $e$. We will assume this is the case from here on and write $S=T^{-1}$. An interesting property of box-world is that all effects lie in the convex cone whose extreme rays are the fiducial effects. It can easily be checked that if $\trans$ is reversible, it must map this cone onto itself and hence map fiducial effects to fiducial effects. Moreover if $E=\sum_{\alpha} e_{\alpha}$ then $\trans(E) = \sum_{\alpha} \trans(e_{\alpha})$. Clearly, $\trans$ also maps multiform effects to multiform effects. 

\subsection{Recovering the identical systems case.} Using the results of \S 3, it is relatively easy to recover the main result of \cite{colbeck10}: that all reversible transformations on Boxworld are trivial as long as all the systems are identical. In fact, we only require that $K^{(i)}_1 = K^{(1)}_1$ for all systems $i$. 

The action of reversible transformations on sub-unit effects forms a crucial step in the proof. Let $E$ be a $i$-sub-unit effect for some $i$, and note that $E$ has a decomposition with exactly $r=K^{(1)}_1$ elements (according to a measurement on system $i$ with $r$ outcomes). It follows that $\trans(E)$ is a multiform effect with at least one decomposition with exactly $K^{(1)}_1$ elements, hence by Corollary \ref{mindecompcor}, $\trans(E)$ is a sub-unit effect.

Since reversible transformations permute the set of sub-unit effects, they preserve the property that a pair of fiducial effects differs in only one component: if $e_1$ and $e_2$ differ in one component, then they each belong to some decomposition of a sub-unit effect $E$.  It follows that $\trans(e_1)$ and $\trans(e_2)$ each belong to some decomposition of the sub-unit effect $\trans(E)$, hence by Lemma \ref{redefflemma} differ in only one component.

Let $\alp_i$ be the set $\{X^{(i)}_{a|x}\}$ of fiducial effects on system $i$. By viewing the fiducial effect $X_{a_1|x_1}^{(1)}\otimes \cdots \otimes X_{a_N|x_N}^{(N)}$ as a string $\textbf{X} = (X_{a_1|x_1}^{(1)}, \ldots , X_{a_N|x_N}^{(N)}) $, $\trans$ can be interpreted as a permutation of the set $\alp_1 \times \cdots \times \alp_N$. The \emph{Hamming distance} between two strings $\textbf{X}_1$, $\textbf{X}_2$ is the number of components in which they differ, i.e. $\# \{ i : \textbf{X}_1^{(i)} \neq \textbf{X}_2^{(i)}\}$. By the above comments, $\trans$ preserves a Hamming distance of 1 between pairs of strings. The following lemmas and resultant theorem are proved in \cite{colbeck10}. Importantly, the proofs do not assume identical systems or equal numbers of outcomes for local fiducial measurements.

\begin{lemma}[Proved in \cite{colbeck10}]\label{hamdistlemma}
Let $\alp_1,\ldots,\alp_N$ be finite alphabets, and $Q$ be a bijective map from $\alp_1\times\ldots\times \alp_N$ to itself. If $Q$ preserves a Hamming distance of 1 between pairs of strings, then it is a composition of operations which permute components, followed by local permutations acting on individual components.
\end{lemma}

\begin{lemma}[Proved in \cite{colbeck10}]\label{locopslemma}
The only reversible transformations allowed on single Boxworld systems are relabellings of measurement choices, and relabellings of measurement outcomes. 
\end{lemma}

\begin{theorem}[Proved in \cite{colbeck10}]
The only reversible transformations allowed in Boxworld with identical non-classical systems are permutations of systems, followed by local relabellings of measurement choices and measurement outcomes.
\end{theorem}

\subsection{Extending to the general case.} We now relax the condition that $K^{(i)}_1 = K^{(1)}_1$ for all systems $i$. This makes the task more complicated, however the method of proof is still to argue that $\trans$ permutes the set of sub-unit effects. 

\begin{lemma}\label{mapredefflemma}
Reversible Boxworld transformations map sub-unit effects to sub-unit effects.
\end{lemma} 

\begin{proof} See appendix \ref{app_lem5} for a detailed proof. However, we will sketch the main ideas of the proof here.

If $K^{(i)}_1 > K^{(1)}_1$, then sub-unit effects at system $i$ will not have decompositions with $ K^{(1)}_1$ elements, so it is not possible to directly apply Corollary \ref{mindecompcor}. However, Corollary \ref{mindecompcor} \emph{can} still be applied for every system $i$ with  $K^{(i)}_1 = K^{(1)}_1$, such that $\trans$ permutes the set of sub-unit effects within this set of systems.

Now consider a system $j$ for which $K^{(j)}_1$ is the next possible greater value than $K^{(1)}_1$. A $j$-sub-unit effect $E$ must be transformed to something which is multiform with some decomposition $\{e_{\alpha}\}$ involving $K^{(j)}_1$ elements. It turns out that if (a) $\sum_{\alpha} e_{\alpha}$ is not a sub-unit effect, and (b) $\trans$ permutes sub-unit effects on systems with $K^{(i)}_1 = K^{(1)}_1$, then for any other decomposition $\{f_{\beta}\}$ of $\trans(E)$ there exists a pure product state $s$ which hits $f_1$ but none of the effects $\{e_{\alpha}\}$.

This sets up an iterative process, at each stage assuming that $\trans$ permutes the sub-unit effects on systems with smaller numbers of outcomes. The iteration terminates when $K^{(j)}_1$ takes on its maximal value, and $\trans$ thus permutes the complete set of sub-unit effects in the multipartite system. \end{proof} 

Lemma \ref{mapredefflemma} implies that $\trans$ preserves a Hamming distance of 1 between effects in the general case, so we can apply Lemmas \ref{hamdistlemma} and \ref{locopslemma} to deduce that reversible transformations are permutations of systems, followed by local relabellings. However, not every permutation of systems is now possible. We say that systems are of the same \emph{type} if they have the same number of fiducial measurement choices, and the fiducial measurements from each system can be matched up so that paired measurements have the same number of outcomes.

\begin{lemma} \label{permopslemma}
Let $\trans$ be an allowed reversible Boxworld transformation which is a permutation of systems $P$, followed by a composition of local relabellings $Q$. Then $P$ can only permute systems of the same type.
\end{lemma}

\begin{proof}
Suppose that $P$ takes system $i$ to system $i'$. Because $Q$ is constant on any sub-unit effect and $\trans$ permutes sub-unit effects (by Lemma \ref{mapredefflemma}), $P$ also permutes sub-unit effects.

Any two fiducial effects $e_1$, $e_2$ which differ only by measurement outcome on system $i$, belong to the same decomposition of some $i$-sub-unit effect $E$. $P(e_1)$ and $P(e_2)$ therefore belong to the same decomposition of $P(E)$, which by the above reasoning is an $i'$-sub-unit effect for some $i'$. It follows from Lemma \ref{redefflemma} that $e_1$ and $e_2$ differ only by measurement outcome on system $i'$.

If instead $e_1$, $e_2$ differ by measurement choice on system $i$, then they belong to \emph{distinct} decompositions of $E$ (as above), hence $P(e_1)$ and $P(e_2)$ belong to different decompositions of $P(E)$, therefore differ by measurement choice on system $i'$.

Hence $P$ must map the fiducial effects on system $i$ onto the fiducial effects on system $i'$ in such a way that effects belong to the same measurement choice after the transformation iff they also did before the transformation. The only way this is possible is if the systems are of the same type.
\end{proof}

\begin{theorem}
The only reversible transformations of non-classical systems allowed in Boxworld are permutations of systems of the same type, followed by local relabellings of measurement choices and measurement outcomes.
\end{theorem}

\begin{proof}
To complete the proof, we need only check that all transformations of the above form are allowed. This is obvious from considering the action on distributions $\condprob{a_1 , \ldots , a_N}{x_1 , \ldots , x_N}$.
\end{proof}

\section{ Discussion. } We have refined and extended the result of \cite{colbeck10}, and demonstrated that -- as long as no subsystem is classical -- reversible dynamics of arbitrary joint Boxworld systems always take the form of permutations of subsystems of the same type, followed by relabellings of measurement choices and outcomes on individual subsystems. If any single subsystem is classical, the result immediately fails: any local operation on another subsystem which is conditioned on a particular outcome on the classical system, is allowed.

This places Boxworld in stark contrast to quantum theory and nature itself, in which subsystems are clearly able to interact in a reversible way, commonly becoming entangled via continuous and reversible processes. In fact quantum theory obeys a much stronger principle than reversibility: any set of perfectly distinguishable states of a system can be mapped via unitary evolution to any other set of the same size. The importance of reversibility in quantum theory is suggested by the prevalent usage of it (or stronger versions of it) as an axiom in numerous information-theoretic derivations of quantum theory \cite{hardy01, brukner09, hardy11, masanes11, chiribella11}. 

It is worth comparing our proof with the proof in \cite{colbeck10} in a little more detail. In \cite{colbeck10} the authors employ a specific (albeit natural) choice of vectors to represent fiducial effects, for which it turns out that all reversible transformations correspond to orthogonal maps. Since the number of measurement choices $M$ and outcomes $K$ are constant across subsystems in their analysis, the same vector space representation can be chosen for every subsystem, and the fact that $\trans$ preserves inner products can be exploited.

On first impressions, one might think that this method  of proof extends to the case where not all systems are identical. However, it can be shown that it breaks down for joint systems with carefully chosen, differing values of $M$ and $K$. There also seems to be no obvious way to embed these pathological cases into larger Boxworld systems in such a way as to inherit their dynamical properties. 

Our method of proof exploits only the linear structure inherent in any representation, and perhaps as a consequence extends more readily to the whole of Boxworld, and may be easier to extend to other examples of general probabilistic theories. An interesting question is whether our results extend to any theory  whose joint states are generated by the maximal tensor product of the local state spaces (i.e. which admit any joint state  consistent with local fiducial measurements) \cite{barnum06}. Such theories share the property of Boxworld that all effects lie in a convex cone  whose extreme rays are product effects which must be permuted by any reversible transformation. However, the structure of the local effects will generally be more complex than in Boxworld, with the extreme rays of the local effect cone differing from the local fiducial effects. An interesting next step would be to investigate systems whose local state spaces are polygons  \cite{janotta11} and which have the maximal tensor product to see if they have trivial reversible dynamics. 

At present the authors do not know of any theory which obeys reversibility which is not just a subset of quantum theory \cite{Note1}. An interesting conjecture is that any theory in the general probabilistic framework which satisfies local tomography and reversibility can be represented within quantum theory. If this were the case, it would suggest that the reversibility of physcal laws lies at the heart of the obscure mathematical structure of quantum theory.  Alternatively, any counterexample to this conjecture would be a fascinating foil theory with which to compare quantum theory, and perhaps to move beyond it.

\bigskip

\acknowledgments  AJS thanks the Royal Society for their support. SWA is funded by an EPSRC grant.

\appendix

\setcounter{lemma}{0}

\section{} \label{app_lem1}
\begin{lemma} 
Let $E = \sum_{\alpha} e_{\alpha}$ be an $i$-sub-unit effect. Then each $e_{\alpha}$ differs from $E$ only in component $i$. Moreover, the set of $i$th components $\{e^{(i)}_{\alpha}\}$ forms a fiducial measurement on system $i$.
\end{lemma}

\begin{proof}  We will prove the lemma by contradiction.  Suppose first that $e_{\alpha'}^{(j)}\neq E^{(j)}$ for some $\alpha'$ and some $j\neq i$. Let $E^{(j)} = X^{(j)}_{a_j|x_j}$ and $e_{\alpha'}^{(j)} = X^{(j)}_{a'_j|x'_j}$. Either $x_j \neq x'_j$, or $x_j = x'_j$ but $a_j \neq a'_j$, so we can construct a pure product state $s^{(j)}$ on system $j$ such that $s^{(j)}_{x_j} \neq a_j$ and $s^{(j)}_{x'_j} = a'_j$, i.e. which hits $e_{\alpha'}$ but not $E$. Then for any pure product state $s$ whose $j$th component is $s^{(j)}$, $\innprod{E}{s} = 0$ but $\innprod{e_{\alpha'}}{s} = 1$, contradicting the fact that  $\innprod{e_{\alpha'}}{s} \leq \innprod{E}{s}$.

$\{ e^{(i)}_{\alpha} \}$ is a set of fiducial effects satisfying $\sum_{{\alpha}} e^{(i)}_{\alpha} = \unit^{(i)}$, hence any pure state $s^{(i)}$ on system $i$ must hit exactly one of the $e^{(i)}_{\alpha}$. If any two of the $e^{(i)}_{\alpha}$ are effects corresponding to different measurements, then there is a pure state $s^{(i)}$ which hits both of them. Hence the effects all belong to the same fiducial measurement $x$; if $\{e^{(i)}_{\alpha}\}$ is not the \emph{full} set of outcomes of measurement $x$, then there is a pure state $s^{(i)}$ which hits none of them. It follows that $\{e^{(i)}_{\alpha}\}$ forms a fiducial measurement on system $i$. 
\end{proof}

\section{} \label{app_lem2}
\begin{lemma} 
For $r\leq K^{(1)}_1$ suppose that $\{e_{\alpha}\}_{\alpha=1}^{r}$does not cover any sub-unit effects. Then for any fiducial effect $f \notin \{e_{\alpha}\}$, there is a pure product state which hits $f$ but none of the $e_{\alpha}$.
\end{lemma}
\begin{proof}[Proof.]
Let $f = X_{a_1|x_1}^{(1)}\otimes \cdots \otimes X_{a_N|x_N}^{(N)}$. We proceed by induction on the number of systems $N$. When $N=1$ set $s_{x_1} = a_1$ to ensure that $s$ hits  $X_{a_1|x_1}^{(1)}$. The conditions imply that no partial sum of $\{e_{\alpha}\}$ equals the unit effect, hence for each other choice of measurement $x' \neq x_1$, it must be possible to choose $s_{x'}$ such that $X_{s_{x'}|x'} \notin \{e_{\alpha}\}$. By construction $s$ hits $X_{a_1|x_1}^{(1)}$ but none of the $e_{\alpha}$.

When $N>1$, note that for any fiducial effect $g$ on system 1, the set $\{e^{\{2,\ldots , N\}}_{\alpha} : e^{(1)}_{\alpha} = g\}$ is a decomposition of some  effect on the remaining $N-1$ systems with at most $K^{(1)}_1 \leq K^{(2)}_1$ elements. This decomposition satisfies the conditions of the Lemma, hence by induction there exists a pure product state $s^{(2)} \otimes \cdots \otimes s^{(N)}$ which hits $f^{\{2,\ldots N\}}$ but none of the set $\{e^{\{2,\ldots , N\}}_{\alpha} : e^{(1)}_{\alpha} = g\}$. 

Again, it is necessary to set $s^{(1)}_{x_1}=a_1$. Consider the set $\{e_{\alpha}^{(1)}\}$, and the outcomes for measurements other than $x_1$ on system 1. One of two cases must occur:

\begin{itemize}

\item[(a)] The set $\{e_{\alpha}^{(1)}\}$ fills none of the other measurements, i.e. for every $x' \neq x_1$, there is an $a_{x'}$ such that $X_{a_{x'}|x'}^{(1)} \notin \{e^{(1)}_{\alpha}\}$. For each such $x'$ set $s_{x'}^{(1)} = a_{x'}$ so that $s$ can hit $e_{\alpha}$ only if $e_{\alpha}^{(1)} = f^{(1)}$. However, using the inductive hypothesis, there exists a pure product state $s^{(2)} \otimes \cdots \otimes s^{(N)}$ which hits $f^{\{2,\ldots ,N\}}$ but none of the set $\{e^{\{2,\ldots , N\}}_{\alpha} : e^{(1)}_{\alpha} = f^{(1)}\}$. 

\item[(b)] There exists a measurement $x'\neq x_1$ on system 1 with $K^{(1)}_{x'}=r$ which is filled by the set $\{e_{\alpha}^{(1)}\}$, i.e. (after reordering) $e_{\alpha}^{(1)} = X_{{\alpha}|x'}^{(1)}$ for $1 \leq {\alpha} \leq r$. $\{e_\alpha\}$ covers no sub-unit effects, so there must be some ${\alpha}'$ and some system $i \neq 1$ such that $e^{(i)}_{{\alpha}'} \neq f^{(i)}$. Set $s^{(1)}_{x'} = {\alpha}'$ so that $s$ does not hit any $e_{\alpha}$ with ${\alpha} \neq {\alpha}'$; the remaining components of $s^{(1)}$ may be chosen arbitrarily. By the inductive hypothesis, there exists a pure product state $s^{(2)} \otimes \cdots \otimes s^{(N)}$ which hits $f^{\{2,\ldots ,N\}}$ but not the single effect $e^{\{2,\ldots ,N\}}_{{\alpha}'}$.

\end{itemize}

In both cases, by construction $s = s^{(1)} \otimes \cdots \otimes s^{(N)}$ hits $f$ but none of the $e_{\alpha}$.
\end{proof}

\section{} \label{app_lem5}

It will be convenient to let $\subuniteffsati$ denote the set of sub-unit effects at system $i$, and for a subset of systems $\Omega \subseteq [N]$, let $\mathcal{S}_{\Omega} = \cup_{i \in \Omega} \subuniteffsati$. The proof of Lemma \ref{mapredefflemma} relies on some simple results about reversible transformations.

\setcounter{lemma}{5}

\begin{lemma} \label{diffoneefflemma}
Let $\Omega \subseteq [N]$ and suppose that $\trans$  is an allowed reversible transformation which permutes the set $\mathcal{S}_{\Omega}$. Then the tranformations of two fiducial effects will be identical outside $\Omega$ if and only if the original effects were. i.e.:
\begin{align*}
e_1^{(k)} = e_2^{(k)}\; \forall k \in [N] \backslash \,\Omega \iff \trans(e_1)^{(k)} = \trans(e_2)^{(k)} \\
\forall k \in [N] \backslash \,\Omega
\end{align*}
\end{lemma}

\begin{proof}
Suppose firstly that the fiducial effects $e_1$ and $e_2$ differ only in one component $i \in \Omega$. $e_1$ and $e_2$ belong to (possibly different) decompositions of a unique sub-unit effect $E \in \subuniteffsati$. By assumption $\trans(E)$ is an $i'$-sub-unit effect for some $i' \in \Omega$; $\trans(e_1)$ and $\trans(e_2)$ belong to decompositions of $\trans(E)$, hence by Lemma \ref{redefflemma} can only differ in component $i'$.

Suppose now that $e_1$, $e_2$ satisfy $e_1^{(k)} = e_2^{(k)}$ for all $k \notin \Omega$, but that they differ in any number of components belonging to $\Omega$. Then it is possible to move from $e_1$ to $e_2$ by changing one component at a time (each component belonging to $\Omega$). At each step, $\trans$ maps the corresponding pair of effects to a pair which differ only in components belonging to $\Omega$. Hence  $\trans(e_1)^{(k)} = \trans(e_2)^{(k)}$ for all $k \notin \Omega$.

To prove the converse direction, note that if $\trans$ is an allowed reversible transformation which permutes the set $\mathcal{S}_{\Omega}$, then so is $(\trans)^{-1}$.
\end{proof}

\begin{lemma} \label{inversesubunitlemma}
Suppose that $E=\sum_{\alpha=1}^{r} e_{\alpha}$ is a sub-unit effect, and that $\{\trans(e_{\alpha})\}$ covers a sub-unit effect $F$. Then $\trans(E) = F$.
\end{lemma}

\begin{proof}
Without loss of generality let $\sum_{\alpha=1}^{s} \trans(e_{\alpha}) = F$ for $s \leq r$, and let $\sum_{\beta} f_{\beta}$ be a distinct decomposition of $F$. Then $E$ covers the multiform effect $(\trans)^{-1}(F) = \sum_{\alpha = 1}^{s} e_{\alpha} = \sum_{\beta = 1} (\trans)^{-1}(f_{\beta})$. It follows from Corollary \ref{redeffcor} that $s=r$ and that $\trans(E) = F$.
\end{proof}

\begin{lemma} \label{cover-unit-lemma}
Suppose that $\{e_{\alpha}\}$ does not cover any sub-unit effects, but that there exists some system $i$ for which $\sum_{\alpha} e_{\alpha}^{(i)} = \unit^{(i)}$. Then for any fiducial effect $f \notin \{e_{\alpha}\}$, there exists a pure product state which hits $f$ but none of the $e_{\alpha}$.
\end{lemma}

\begin{proof}
Let $f^{(i)} = X^{(i)}_{a|x}$ and $\Omega_{i} = [N]\setminus \{i\}$. Note that $\{e_{\alpha}^{(i)}\}$ is the complete set of outcomes for some fiducial measurement $x'$ on system $i$: without loss of generality, $e_{\alpha}^{(i)} = X^{(i)}_{\alpha | x'}$. 

If $x'=x$, then $f^{(i)} = e^{(i)}_{a}$. Set  $s^{(i)}_{x} = a$ and choose the remaining components of $s^{(i)}$ arbitrarily, so that $s^{(i)}$ hits $e^{(i)}_{a}$ but none of the other $e^{(i)}_{\alpha}$. Note that $f^{\Omega_{i}}$ and $e^{\Omega_{i}}_{a}$ must be distinct fiducial effects, so by Lemma 2 there exists a pure product state $s^{\Omega_{i}}$ which hits $f^{\Omega_{i}}$ but not $e^{\Omega_{i}}_{a}$.

If $x'\neq x$, then since $\sum e_{\alpha}$ is not a sub-unit effect, there exists ${\alpha}'$ and $i' \neq i$ such that $e_{{\alpha}'}^{(i')} \neq f^{(i')}$. Set $s^{(i)}_{x} = a$ and $s^{(i)}_{x'} = {\alpha}'$, and choose the remaining components of $s^{(i)}$ arbitrarily. By Lemma \ref{mindecomplemma} there is a pure product state $s^{\Omega_i}$ which hits $f^{\Omega_i}$ but not the single fiducial effect $e_{{\alpha}'}^{\Omega_i}$.

In both cases, combining $s^{(i)}$ with $s^{\Omega_{i}}$ gives a pure product state $s$ which hits $f$ but none of the $e_{\alpha}$.
\end{proof}

We are now in a position to prove Lemma \ref{mapredefflemma} from the main text. 

\setcounter{lemma}{4}

\begin{lemma}
Reversible Boxworld transformations map sub-unit effects to sub-unit effects,  so long as none of the systems are classical. 
\end{lemma} 
\begin{samepage}
\begin{proof}
We begin by considering the action of $\trans$ on a $1$-sub-unit effect $E$. $\trans(E)$ is a multiform effect with a decomposition containing $K^{(1)}_1$ elements, hence by Corollary \ref{mindecompcor} it is a $i$-sub-unit effect for some system $i$ with $K^{(i)}_1 = K^{(1)}_1$. By the same reasoning $\trans$ permutes the set $\mathcal{S}_{\Omega}$, where $\Omega = \{i : K_1^{(i)} = K_1^{(1)}\}$.

We now show iteratively that $\trans$ permutes the sub-unit effects at systems with $r=K^{(i)}_1 > K^{(1)}_1$. Let $\sum_{{\alpha}=1}^{r} e_{\alpha} = \sum_{{\beta}=1}^{s} e'_{\beta}$ be distinct decompositions of an $i$-sub-unit effect $E$ and assume that $\trans$ permutes the set $\mathcal{S}_{\Omega}$, where $\Omega = \{j : K_1^{(j)} < r\}$. Note that $\trans(E)$ is also multiform, since $\trans(E) = \sum_{{\alpha}=1}^{r} \trans(e_{\alpha}) = \sum_{\beta} \trans(e'_{\beta})$. Write $f_{\alpha}=\trans(e_{\alpha})$ and $g=\trans(e'_1)$, noting that $g \notin \{f_{\alpha}\}$.

Assuming that $\{f_{\alpha}\}$ does not cover a sub-unit effect, our aim is to construct a pure product state $s$ that hits $g$ but none of the $f_{\alpha}$, giving a contradiction. Hence $\{\trans(e_{\alpha})\}$ must cover a sub-unit effect. It then follows from Lemma \ref{inversesubunitlemma} that $\trans(E)$ is itself an $i'$-sub-unit effect for some system $i'$ with $ K_1^{(i')} = K_1^{(i)}$. By iteration we can then complete the proof of the lemma. 

To obtain the contradiction mentioned above, suppose firstly that $\{f_{\alpha}^{(i)}\}_{\alpha=1}^r$ covers $\unit^{(i)}$. Then because all measurements on system $i$ have at least $r$ outcomes it must be the case that  $\sum_{\alpha} f_{\alpha}^{(i)} = \unit^{(i)}$, and by Lemma \ref{cover-unit-lemma} there exists a state which hits $g$ but none of the $f_{\alpha}$. 

Suppose instead that $\{f_{\alpha}^{(i)}\}_{\alpha=1}^r$ does not cover $\unit^{(i)}$. Then for every measurement  $x' $ on system $i$ there is an outcome $a'$ such that $X^{(i)}_{a'|x'} \notin \{f_{\alpha}^{(i)}\}$. Let $g^{(i)} = X_{a|x}^{(i)}$ and set $s^{(i)}_{x}=a$. Then for each $x' \neq x$ set $s^{(i)}_{x'} = a'$, so that $s$ cannot hit $f_{\alpha}$  if $f_{\alpha}^{(i)} \neq g^{(i)}$. Let $\bar{\Omega} = [N] \backslash (\Omega \cup \{i\})$, and consider the set $\{f_{\alpha}^{\bar{\Omega}} : f_{\alpha}^{(i)} = g^{(i)}\}$ and the effect $F^{\bar{\Omega}}$ which admits this decomposition. This set has size at most $r$ and by Lemma \ref{diffoneefflemma} does not contain $g^{\bar{\Omega}}$,  as $e_1'$ differs from each member of $\{e_\alpha\}$  on system $i$ (outside $\Omega$), and hence all the $\{f_{\alpha}\}$ must differ from $g$ outside  $\Omega$.

If $\{f_{\alpha}^{\bar{\Omega}} : f_{\alpha}^{(i)} = g^{(i)}\}$  does not cover a sub-unit effect, then by Lemma \ref{mindecomplemma} there exists a pure product state $s^{\bar{\Omega}}$ which hits $g^{\bar{\Omega}}$ but none of the $f_{\alpha}^{\bar{\Omega}}$. Combining $s^{\bar{\Omega}}$ with $s^{(i)}$ and any choice of $s^{\Omega}$ that hits $g^{\Omega}$ gives a pure product state $s$ which hits $g$ but none of the $f_{\alpha}$.


If $\{f_{\alpha}^{\bar{\Omega}} : f_{\alpha}^{(i)} = g^{(i)}\}$ covers an $i'$-sub-unit effect for some $i' \in \bar{\Omega}$, then  it must be the case that $\sum_{\alpha}  f_{\alpha}^{(i')} = \unit^{(i')}$ because all measurements on systems in $\bar{\Omega}$ have at least $r$ outcomes. It follows from Lemma \ref{cover-unit-lemma} that there exists a state which hits $g$ but none of the $f_{\alpha}$.

\end{proof}
\end{samepage}

\end{document}